\documentclass[journal]{IEEEtran}

\usepackage{latexsym,amsmath,amssymb}
\def\F {{\mathbb{F}}}

\def\Fq {{\mathbb{F}_q}}
\def\C{\mathbb{C}}

\def\R{\mathbb{R}}
\def\cX{{\mathcal X}}
\def\X{{\mathcal X}}

\def\cL{{\mathcal L}}

\def\GRS{{\mathcal GRS}}
\def\ba{{\bf a}}
\def\bb{{\bf b}}
\def\bc{{\bf c}}

\def\res{{\rm res}}
\def\Gal{{\rm Gal(\overline{\F}_q/\F_q)}}
\def\cP{{\mathcal P}}
\def\bu{{\bf u}}
\def\bv{{\bf v}}

\def\bo{{\bf 0}}
\def\bi{{\bf 1}}

\def\bx{{\bf x}}

\def\a{{\alpha}}

\def\beq{\begin{equation}}
\def\eeq{\end{equation}}

\newtheorem{thm}{Theorem}[section]

\newtheorem{prop}[thm]{Proposition}
\newtheorem{lem}[thm]{Lemma}

\numberwithin{equation}{section} \newtheorem{rem}[thm]{Remark}
\newtheorem{exm}[thm]{Example}

\begin{document}

\title{Euclidean and Hermitian Self-orthogonal Algebraic Geometry Codes and Their Application to Quantum Codes}

\author{Lingfei~Jin and Chaoping~Xing
\thanks{L.F.Jin  and C. P. Xing are with Division of
Mathematical Sciences, School of Physical and Mathematical Sciences,
Nanyang Technological University, Singapore 637371, Republic of
Singapore (email:\{lfjin, xingcp\}@ntu.edu.sg).}

\thanks{The work is partially supported by Singapore MOE
 Tier 2 research grant T208B2906 and Singapore MOE
 Tier 1 research grant RG60/07.}}

\maketitle

\begin{abstract}
In the present paper, we show that if the dimension of an arbitrary algebraic geometry code over a finite field of even characters is slightly less than half of its length, then it is equivalent to an Euclidean self-orthogonal code. However, in the literatures, a strong contrition about existence of certain differential is required to obtain such a result. We also show a similar result on Hermitian self-orthogonal algebraic geometry codes. As a consequence, we can apply our result to quantum codes and obtain quantum codes with good asymptotic bounds.
\end{abstract}


\section{Introduction}
Classical Euclidean self-orthogonal codes have been extensively studied due to their nice algebraic and combinatorial nature \cite{Mac Slo1, Mac Slo}. Various constructions of classical Euclidean self-orthogonal codes have been studied through algebraic and combinatorial tools \cite{Cal Rai, Che Lin Xin, Lin Ling Luo Xin}. In recent years, this topic has become increasingly interesting due to application to quantum codes \cite{Aly Kla Sar, Ash Kni, Bie Ede, Gra Bet, Hu Tan, Ket Kla}. For application to quantum codes, one is interested in not only  classical Euclidean self-orthogonal codes but also some other types of self-orthogonal codes such as Hermitian and simplectic self-orthogonal codes.

One good candidate for self-orthogonal codes is algebraic geometry codes. For instance, in \cite{St}, it is shown that algebraic geometry codes from a certain optimal tower are equivalent to Euclidean self-orthogonal codes. Unfortunate, this is not true for an arbitrary algebraic geometry codes in general. In fact, it requires a very strong condition in order that an algebraic geometry code is Euclidean self-orthogonal.

In this paper, we construct both  Euclidean and Hermitian self-orthogonal codes through algebraic geometry codes. More precisely, we show that an arbitrary algebraic geometry code with dimension slightly less than half of its length over a finite field of characteristic $2$ is Euclidean self-orthogonal. Furthermore, it is shown that an arbitrary algebraic geometry code with dimension slightly less than half of its length over a finite field of arbitrary characteristic is Hermitian self-orthogonal when its tensor product is considered (see the details in ).

The paper is organized as follows.

\section{Preliminary}
To construct self-orthogonal algebraic geometry codes, we need to recall some basic definition and results of algebraic curves and algebraic geometry codes. The reader may refer to \cite{St93, Ts Vl}

Let $\cX$ be a smooth, projective, absolutely irreducible curve of genus $g$ defined over $\Fq$. We denote by $\Fq(\cX)$ the function field of $\cX$. An element of $\Fq(\cX)$ is called a function. The normalized discrete valuation corresponding to a point $P$ of $\Fq(\cX)$ is written as $\nu$. A point $P$ is said $\F_q$-rational if $P^{\sigma}=P$ for all $\sigma$ in the Galois group $\Gal$. Likewise, a divisor $G=\sum_{P}m_pP$ is said $\F_q$-rational if $G^{\sigma}=\sum_Pm_PP^{\sigma}=G$ for all $\sigma$ in the Galois group $\Gal$.

For an $\F_q$-rational divisor $G$, the Riemann-Roch space associated to $G$ is
\[\cL_{\Fq}(G)=\{f\in \Fq(\cX):{\rm div}(f)+G\geq0\}\cup\{0\}\]
Then $\cL_{\Fq}(G)$ is a finite-dimensional vector space over $\Fq$ and we denote its dimension by $\ell(G)$. By the Riemann-Roch theorem we have
\[\ell(G)\geq \deg(G)+1-g\]
where the equality holds if $\deg(G)\geq 2g-1$.

We can also consider the tensor product of $\cL_{\Fq}(G)$ with $\F_{q^2}$, denoted by $\cL_{\F_{q^2}}(G)$, i.e.,{\small
\[\cL_{\F_{q^2}}(G)=\cL_{\Fq}(G)\otimes_{\Fq}\F_{q^2}=\{f\in \F_{q^2}(\cX):{\rm div}(f)+G\geq0\}\cup\{0\}.\]}
Then $\cL_{\F_{q^2}}(G)$ is a vector space over $\F_{q^2}$ of dimension $\ell(G)$.

Let $P_1,\dots,P_n$ be pairwise distinct $\F_q$-rational points of $\cX$ and $D=P_1+\dots+P_n$. Choose an $\F_q$-rational divisor $G$ in $\cX$ such that ${\rm supp}( G)\bigcap {\rm supp}( D)=\varnothing$, and a vector $\bv=(v_1,\dots,v_n)$ such that $v_i\in(\Fq)^*, (i=1,\dots,n)$ . Then $\nu_{P_i}(f)\geq0$ for all $1\leq i\leq n$ and any $f\in \cL_{\Fq}(G)$.

Consider the map
\[\Psi:\cL(G)\rightarrow \F_q^n,\quad f\mapsto(v_1f(P_1),\dots,v_nf(P_n)).\]
Obviously the image of the $\Psi$ is a subspace of $\F_q^n$. The image of $\Psi$ is denoted as $C_\cL(D,G)$ which is called an algebraic-geometry code(or AG code for short). If $\deg(G)<n$, then $\Psi$ is an embedding and we have $\dim(C_\cL(D,G))=\ell(G)$. By Riemann-Roch theorem we can estimate the parameters of an AG code (see \cite{St93}).

\begin{prop}\label{2.1} $C_\cL(D,G,\bv)$ is an $[n,k,d]-$ linear code over $\Fq$ with parameters
\[k=\ell(G)-\ell(G-D),\quad  d\geq n-\deg(G).\]
\begin{itemize}
\item[{\rm (a)}] If $G$ satisfies $g\leq \deg(G)<n$, then
\[k=\ell(G)\geq \deg(G)-g+1,\quad  d\geq n-\deg(G).\]
\item[{\rm (b)}] If additionally $2g-2<deg(G)<n$, then $k=deg(G)-g+1$.
\end{itemize}
\end{prop}
\begin{rem}\label{2.2}{\rm \begin{itemize}
\item [(i)] The proposition above implies that $k+d\geq n+1-g$. Compared
with the Singleton bound, we can know all the AG codes in the above
are MDS codes while in rational function field. \item[(ii)] Note
that $C_{\cL}(D,G,\bi)$ is the ordinal algebraic geometry code
defined by Goppa, where $\bi$ denotes the all-one vector
$(1,\dots,1)$.
\end{itemize}}
\end{rem}

Now we discuss the Euclidean dual of the algebraic code $C_\cL(D,G;\bv)$.

For two vectors $\ba=(a_1,\dots,a_n), \bb=(b_1,\dots,b_n)$ in
$\F_q^n$, Euclidean inner product is defined by
$<\ba,\bb>_E=\sum_{i=1}^na_ib_i$. For a linear code $C$ over $\F_q$,
the {\it Euclidean dual} of $C$ is defined by
\[C^{\perp_E}:=\{\bv\in\F_q^n:\;<\bv,\bc>=0\ \forall\ \bc\in C\}.\]

Let $\Omega$ denote the differential space of $\F_q(\cX)$. For an $\F_q$-rational divisor $G$, we define
\[\Omega(G)=\{w\in \Omega:\;{\rm div}(w)\ge G\}\]
and denote the dimension of $\Omega(G)$ by $i(G)$. Then one has the following relationship
\[i(G)=\ell(K-G),\]
where $K$ is a canonical divisor.

 We define the code $C_\Omega(D,G,\bv)$ as
\[C_\Omega(D,G,\bv)=\{(v_1\res_{P_1}(w),\dots,v_n\res_{P_n}(w)):\;w\in \Omega(G-D)\},\]
where $\res_{P_i}(w)$ stands for the residue of $w$ at $P_i$.

 $C_\Omega(D,G,\bv)$ is an $[n,i(G-D)-i(G),\geq \deg G-(2g-2)]$ linear code over $\Fq$.
Furthermore, $C_{\Omega}(D,G,\bv^{-1})$ is the Euclidean dual of $C_{\cL}(D,G,\bv)$,
where $\bv^{-1}$ denotes the vector $(v_1^{-1},\dots,v_n^{-1})$.

\section{Self-orthogonal algebraic geometry codes}
In this section, we first show existence of a ceratin vector in the
Euclidean dual code $C_{\Omega}(D,G,\bi)$ of $C_{\cL}(D,G,\bi)$.
Based on this result, we are able to  show that any algebraic
geometry codes are equivalent to  Euclidean and Hamming
self-orthogonal codes.
\subsection{A result on algebraic geometry codes}
\begin{prop}\label{3.1} The code $C_\Omega(D, 2G,\bi)$ contains at least a codeword of Hamming weight $n$ if
\[\deg(G)<\frac12\left(n-1-n\log_q\left(1+\frac2q\right)\right).\]
\end{prop}
\begin{proof} Let $m$ denote the degree of $G$.
The number of codewords with Hamming weight $n$ in $C_\Omega(D, 2G,\bv)$ is the size of \[\Omega(2G-\sum_{j=1}^nP_j)\setminus\bigcup_{i=1}^n\Omega(2G-\sum_{j=1}^nP_j+P_i).\]
 We denote this set by $A$ and denote $\Omega(2G-\sum_{j=1}^nP_j+P_i)$ by $A_i$. Thus, it's sufficient to
 prove $A$ is not an empty set. By the inclusion-exclusion principle, we have
 \begin{eqnarray*}
 |A|&=&\left|\Omega(2G-\sum_{j=1}^nP_j)\right|-\sum_{i=1}^n|A_i|+\sum_{h,k}|A_h\cap A_k|\\
 &&+ \dots+(-1)^{n-2m-2g+2}\sum\left|\bigcap_{j=1}^{n-2m+2g-2}A_{i_j}\right|\\
 &=&q^{n-2m+g-1}-{n\choose 1}q^{n-2m+g-2}+{n\choose 2}q^{n-2m+g-3}\\
 &&+\dots+(-1)^{n-2m-1}{n\choose {n-2m-1}}q^g\\
 &&+\sum_{k=n-2m}^{n-2m+2g-2}(-1)^k\sum_{i_1,\dots,i_{k}}\left|\bigcap_{j=1}^{k}A_{i_j}\right|\\
&=&q^{n-2m+g-1}\left(1-\frac{1}{q}\right)^n+c,
\end{eqnarray*}
%Hence, if $n$ is even, we have
%\[A\ge \sum_{k=0}^{n-2m-1}(-1)^k{n\choose k}q^{n-2m+g-1-k}\]
where
 \begin{eqnarray*}c&=&\sum_{k=n-2m}^{n-2m+2g-2}(-1)^k\sum_{i_1,\dots,i_{k}}\left|\bigcap_{j=1}^{k}A_{i_j}\right|\\
 &&-\sum_{k=n-2m}^{n}(-1)^k{n\choose k}q^{n-2m+g-1-k}\\
 &\ge&-\sum_{i_1,\dots,i_{n-2m-1}}\left|\bigcap_{j=1}^{n-2m-1}A_{i_j}\right|\\
 &&-\sum_{k=n-2m}^{n}{n\choose k}q^{n-2m+g-1-k}\\
 &=&-q^{g}\sum_{k=n-2m-1}^{n}{n\choose k}q^{n-2m-1-k}\\
 &\ge&-q^{g}\left(1+\frac1q\right)^n
 .\end{eqnarray*}
The desired result follows from the condition.
\end{proof}

\subsection{Euclidean Self-orthogonal AG Codes}
In this subsection, we  restrict ourselves to finite fields $\Fq$ of
even characteristic. A linear code $C$ is called {\it Euclidean
self-orthogonal} if $<\bu,\bv>_E=0$ for all $\bu,\bv\in C$. It is
clear that the dimension of an Euclidean self-orthogonal code is at
most the half of its length.

\begin{thm}\label{3.2} $C_{\cL}(D,G,\bi)$ is equivalent to an Euclidean self-orthogonal
code if
\[\deg(G)<\frac12\left(n-1-n\log_q\left(1+\frac2q\right)\right).\]
\end{thm}
\begin{proof}
From Proposition \ref{3.1}, there  exists a codeword
$\bu=(u_1,\dots,u_n)$ of Hamming weight $n$ in
$C_{\cL}(D,2G,\bi)^{\perp_E}=C_\Omega(D, 2G,\bi)$. Since $v_i$ are
elements in $\F_q^*$ and $q$ is a power of $2$, there exist $v_i\in\F_q^*$
such that $v_i^2=u_i$ for $i=1,\dots,n$.  For any two codewords $(v_1f(P_1),\dots,v_nf(P_n))$ and $(v_1h(P_1),\dots,v_nh(P_n))$ in $C_{\cL}(D,G,\bv)$ for some
$f,h\in\cL_{\F_q}(G)$, their Euclidean inner product is
\[\sum_{i=1}^nv_i^2f(P_i)h(P_i)=\sum_{i=1}^nu_i(fh)(P_i)=0.\]
Therefore, $C_{\cL}(D,G,\bv)$ is Euclidean self-orthogonal and our result follows.
\end{proof}

\begin{rem}\label{3.3}{\rm In $\deg(G)>2g-2$, then the dimension of $C_{\cL}(D,G,\bi)$ is $\deg(G)-g+1$. Hence, from Theorem \ref{3.2}, an algebraic geometry code is equivalent to an Euclidean self-orthogonal code if its dimension is at most $\frac12\left(n-1-n\log_q\left(1+\frac2q\right)\right)-g$.
}\end{rem}

The following example shows that the condition that $\Fq$ has even characteristic is necessary.

 \begin{exm}\label{3.4} We consider the algebraic code $C_{\cL}(D,G,\bi)$ over  $\F_5$ from the rational function field
  \[\{(f(0),f(1),f(2),f(3)):\; f\in\F_5[x],\; \deg(f)\le 1\},\]
  where the divisors $D$ and $G$ are clear from the above context.
  It is in fact a generalized Reed-Solomon code (see Section IV). It is easy to see that its equivalent code $C(D,G,\bv)$ is Euclidean self-orthogonal if and only if $(v_1^2,v_2^2,v_3^2,v_4^2)$ is a nonzero solution of
 \[
 \left(\begin{array}{cccc}
 1 & 1&1&1\\
 0&1&2&3\\
 0&1&4&4\\\end{array}\right) {\bx}=\bo.\]
On the other hand, all possible nonzero solutions of above equation are $\lambda(2,4,1,3)$ for soem nonzero $\lambda$. However, $2$ and $3$ are non-square elements in $\F_5$ , while $4,1$ are square elements in $\F_5$. This implies that $(v_1^2,v_2^2,v_3^2,v_4^2)$ can not be a nonzero solution, i.e., $C(D,G,\bv)$ is not Euclidean self-orthogonal.
 \end{exm}

\subsection{Hermitian self-orthogonal AG codes}
To study Hermitian self-orthogonal codes, we have to consider codes over $\F_{q^2}$.

For two vectors $\ba=(a_1,\dots,a_n), \bb=(b_1,\dots,b_n)$ in $\F_{q^2}$, we define Hermitian inner product by $<\ba,\bb>_H=\sum_{i=1}^na_ib_i^q$. For an $\Fq$ linear code $C$ in $\F_{q^2}$, the {\it Hermitian dual} $C^{\perp_H}$ of an $\F_q$-linear code $C\subseteq \F_{q^2}^n$  consists of vectors in $\F_{q^2}$ that are orthogonal with all the codewords in $C$ with respect to Hermitian inner product defined above. It follows immediately that $C^{\perp_H}=(C^{q})^{\perp}$, where
$C^q=\{(c_1^q,\dots,c_n^q):\; (c_1,\dots,c_n)\in C\}$. This implies that the Hermitian dual $C^{\perp_H}$ of $C$ is $(C^{q})^{\perp_E}$.

Let $\cX$ be an algebraic curve in $\Fq$, let $P_1,\dots,P_n$ be pairwise distinct $\Fq$-rational points and let $G$ be  an $\F_q$-rational divisor such that ${\rm supp}(G)\cap\{P_1,\dots,P_n\}=\emptyset$. Define a code over $\F_{q^2}$
\[C_{\cL}(D,G,\bv;\F_{q^2}):=\{(v_1f(P_1),\dots,v_nf(P_n)):\; f\in \cL_{\F_{q^2}}(G)\}.\]
Then $C_{\cL}(D,G,\bi;\F_{q^2})$ is an $[n,\ell(G),d\ge n-\deg(G)]$-linear code over $\F_{q^2}$ if $\deg(G)<n$. In fact, $C_{\cL}(D,G,\bi;\F_{q^2})$ is the tensor product  $C_{\cL}(D,G,\bi;\F_{q^2})\otimes_{\F_q}\F_{q^2}$.

\begin{thm}\label{3.5} $C_{\cL}(D,G,\bi;\F_{q^2})$ is equivalent to an Hermitian self-orthogonal
code if
\[\deg(G)<\frac12\left(n-1-n\log_q\left(1+\frac2q\right)\right).\]
\end{thm}
\begin{proof}
From Proposition \ref{3.1}, there  exists a codeword
$\bu=(u_1,\dots,u_n)$ of Hamming weight $n$ in
$C_{\cL}(D,2G,\bi)^{\perp_E}=C_\Omega(D, 2G,\bi)$. Since $v_i$ are
elements in $\F_q^*$, there exist $v_i\in\F_{q^2}^*$
such that $v_i^{q+1}=u_i$ for $i=1,\dots,n$. Moreover, $\bu$ is also Euclidean orthogonal to  $C_{\cL}(D,2G,\bi;\F_{q^2})$ as $C_{\cL}(D,G,\bi;\F_{q^2})$ has a basis from $C_{\cL}(D,2G,\bi)$

Consider two codewords $(v_1f(P_1),\dots,v_nf(P_n))$ and $(v_1h(P_1),\dots,v_nh(P_n))$ in $C_{\cL}(D,G,\bv)$ for some
$f,h\in\cL_{q^2}(G)$. Then $fh^{\sigma}$ is an element of $\cL_{q^2}(G)$, where $\sigma$ is the Frobenius in the Galois group ${\rm Gal}(\overline{\F}_q/\Fq)$.
 their Hermitian inner product is{\small
\[\sum_{i=1}^nv_i^{q+1}f(P_i)(h(P_i))^{q}=\sum_{i=1}^nu_i(fh^{\sigma})(P_i^{\sigma})=\sum_{i=1}^nu_i(fh^{\sigma})(P_i)=0.\]}
Therefore, $C_{\cL}(D,G,\bv;\F_{q^2})$ is  Hermitian self-orthogonal and our result follows.
\end{proof}

\begin{rem}\label{3.6}{\rm To show Euclidean self-orthogonality of $C_{\cL}(D,G,\bi)$, we requires that $\F_q$ has even characteristics. However, we do not need this condition for Hermitian elf-orthogonality of $C_{\cL}(D,G,\bi;\F_{q^2})$.
}\end{rem}

\section{Examples}
In this section, we illustrate our result by considering algebraic codes from projective line and elliptic curves.
\subsection{Self-orthogonal generalized Reed-Solomn codes}

Let's recall some basic results of generalized Reed-Solomon codes ($\GRS$ codes for short) first. Let $\F_q$ be a finite field of $q$ elements, choose $n$ distinct elements $\a_1,\dots,\a_n$ of $\F_q$, and $n$ nonzero
elements $v_1,\dots,v_n$ of $\F_q$. Denote $\ba=(\a_1,\dots,\a_n)$ and $\bv=(v_1,\dots,v_n)$.

Let $P_i$ be the only zero of $x-\a_i$ and let $\infty$ be the only pole of $x$. Put $D=\sum_{i=1}^nP_i$ and $G=(k-1)\infty$ for some $k$ between $1$ and $n$. Then we denote our algebraic geometry codes $C_{\cL}(D,G,\bv)$ and $C_{\cL}(D,G,\bi;\F_{q^2})$ by $\GRS_k(\ba,\bv)$ and $\GRS_k(\ba,\bv,\F_{q^2})$, respectively.
First of all, the Euclidean dual code of $C_{\cL}(D,2G,\bv)=\GRS_{2k-1}(\ba,\bv)$ is $\GRS_{n-2k+1}(\ba,\bv^{'})$, where $\bv^{'}$ is a
nonzero solution of the following system \begin{equation}\label{eq4.1}
\left(\begin{array}{cccc}
v_1 & v_2&\dots&v_n\\
v_1\a_1&v_2\a_2&\dots&v_n\a_n\\
v_1\a_1^2&v_2\a_2^2&\cdots&v_n\a_n^2\\
\cdot&\cdot&\cdots&\cdot\\
\cdot&\cdot&\cdots&\cdot\\
\cdot&\cdot&\cdots&\cdot\\
v_1\a_1^{n-2}&v_2\a_2^{n-2}&\cdots&v_n\a_n^{n-2}\end{array}
\right)\bx^T=\bo.\end{equation}
Note that the solution space of the above system has dimension
$1$ and every nonzero solution has all coordinates not equal to
zero. It is clear that $\GRS_{n-2k+1}(\ba,\bv^{'})$ has a codeword $f(P_1),\dots,f(P_n))$ of weight $n$ for an irreducible polynomial $f$ of degree $2$ as long as $n-2k\ge 2$. Therefore, a generalized Reed-Solomon code $\GRS_k(\ba,\bi)$ is equivalent to an Euclidean self-orthogonal
code if $k\le n-1$.

With the same arguments, we can show that a generalized Reed-Solomon code $\GRS_k(\ba,\bi,\F_{q^2})$ is equivalent to a Hermitian self-orthogonal
code if $k\le n-1$.

\subsection{Codes over elliptic curves}
The AG codes  in the example of $\GRS$ codes are associated with the projective line whose genus is 0. In this subsection, we consider
another example of AG code basing on  elliptic curves.

Let $\cX$ be an elliptic curve over $\F_q$ and let $\cP$ denote the set of $\F_q$-rational points on $\cX$. Choose an $\F_q$-rational divisor $G$ such that ${\rm supp}(G)\cap\cP=\emptyset$. For $2\deg(G)+2\le n<|\cP|$, we choose a closed point $Q$ of degree $n-2\deg(G)$ and $n$ $\F_q$-rational points $P_1,\dots,P_n$ such that $\sum_{i=1}^nP_1-2G+Q$ is equivalent to a canonical divisor $K$ (note that this can be always done). Let ${\rm div}(w)=K$
for a differential $w$ and ${\rm div}(x)=K-(\sum_{i=1}^nP_1-2G+Q)$. Then the differential $xw$ belongs to
$\Omega(2G-D=\sum_{i=1}^nP_i)$. Moreover, $(\res_{P_1}(xw),\dots,\res_{P_n}(xw))$ is a codeword of  $C_{\Omega}(D,2G,\bi)$ with Hamming weight $n$. This implies that $C_{\cL}(D,G,\bi)$ is equivalent to an Euclidean self-orthogonal
code and  $C_{\cL}(D,G,\bi,\F_{q^2})$ is equivalent to a Hermitian self-orthogonal
code.

\section{Application to Quantum codes}

The main purpose of this section is to apply our self-orthogonal codes to construction of quantum codes and derive an asymptotic bound.

Let us first introduce some notations and results on quantum codes. Let $\C$ be the field of complex numbers. For an positive integer $n$, denote $V_n=(\C^{q^n})^{\otimes n}=\C^{q^n}$. Any $K\geq1$ dimensional subspace $Q$ of $V_n$ is called a $q$-ary quantum code with length $n$, dimension $K\geq1$. Then $Q$ is a $((n,K,d))_q$ code or $[[n,k,d]]_q$ code if $Q$ can detect $d-1$ errors and correct $\lfloor\frac{d-1}{2}\rfloor$ where $k=\log_qK$. Similar as the classical code, for any $[[n,k,d]]_q$ quantum code, the quantum singleton bound tells us $n\geq k+2d-2$. $Q$ is called a quantum MDS code if it achieves the quantum singleton bound. In order to use our results to construct quantum code, we need to introduce two lemmas for connection.

\begin{lem}\label{5.1}(see \cite{Ket Kla})There is an $q$-ary $[[n,n-2k,d^{\perp}]]$-quantum code whenever there exists a $q$-ary classical
Euclidean self-orthogonal $[n,k]-$linear code with dual distance $d^{\perp}$.
\end{lem}

\begin{lem}\label{5.2}(see \cite{Ash Kni}) There is an $q$-ary $[[n,n-2k, d^{\perp}]]$-quantum  code whenever there exists a
$q$-ary classical Hermitian self-orthogonal  $[n,k]$-linear code with dual distance $d^{\perp}$.
\end{lem}

 Now, using the theorems in the previous sections, we can derive several classes of quantum codes immediately.
\begin{thm} For finite field $\Fq$ and $1\leq n\leq q+1, k\leq n-1$,
there exists a $q$-ary $[[n,n-2k,k+1]]$-quantum MDS code.
\end{thm}

\begin{thm}  For finite field $\Fq$, $2m+2\leq n\leq q+1+\lfloor2\sqrt{q}\rfloor$, there exists a
$q$-ary $[[n,n-2m,m]]$-quantum code.
\end{thm}
\begin{proof}Applying the result in $\GRS$ codes to Theorem \ref{5.1} yields the desired results.
\end{proof}
\begin{thm} For finite field $\Fq$, $2m+2\leq n<q+1+\lfloor2\sqrt{q}\rfloor$,
there exists a $q$-ary $[[n,n-2m,m]]$-quantum code.
\end{thm}
\begin{proof}Applying the result of elliptic curves to Theorem \ref{5.2} gives the desired results.
\end{proof}

Now, we introduce some results on quantum codes and their asymptotic bounds.
For a  $q$-ary quantum code $Q$, we denote by $n(Q), K(Q)$, and $d(Q)$ the
length, the dimension , and the minimum distance of $Q$, respectively. Let
$U_q^Q$ be the set of ordered pairs $(\delta, R)\in\R^2$ for
which there exists a family $\{Q_i\}_{i=1}^{\infty}$ of  $q$-ary codes  with $n(Q_{i})\rightarrow\infty$ and
\[\delta=\lim_{i\rightarrow\infty}\frac{d(Q_{i})}{n(Q_{i})}, \quad R=\lim_{i\rightarrow\infty}
\frac{\log_qK(Q_{i})}{n(Q_{i})},\]
where $\log_q$ denotes the logarithm to the base $q$. One of the central asymptotic problems for quantum codes is to determine the domain $U_q^Q$. As in  classical coding, it is a hard problem to determine $U_q^Q$ completely. Instead, we are satisfied with some bounds on $U_q^Q$.

 A very good existence lower bound for $p$-ary
quantum codes was introduced by Ashikhmin and Knill \cite{Ask Kni}.
It is called the quantum Gilbert-Varshamov bound.
 As in classical coding theory, the quantum Gilbert-Varshamov
bound is a benchmark for the function $\alpha_q^Q(\delta)$.

For $0<\delta<1$, define the $q$-ary entropy function
$$H_{q}(\delta):=\delta\log_{q}(q-1)-\delta\log_{q}\delta-(1-\delta)\log_{q}
(1-\delta),$$
and put
$$
R_{GV}(q,\delta):=1-\delta\log_q(q+1)-H_{q}(\delta).$$
Then the Gilbert-Varshamov bound says that
\begin{equation}\label{eq3.1}
\alpha_{q}^Q(\delta)\ge R_{GV}(q,\delta) \quad \hbox{for all} \;
\delta\in(0,\frac 12).
\end{equation}

Later on, a bound from algebraic geometry codes was derived in \cite{Chen Lin Xin, Fen Lin Xin, Mat} and this algebraic geometry bound improves the Gilbert-Vrahsamov bound for large $q$ as in the classical case. To introduce the asymptotic algebraic geometry bound, we need some further notations.

For any prime power $q$ and any integer $g\ge 0$, put
$$N_{q}(g):=\max N(\X),$$
where the maximum is extended over all curves $\X/\F_q$ with $g(\X)=g$.

We also define the following asymptotic quantity
$$A(q):=\limsup_{g\rightarrow\infty}\frac{N_{q}(g)}g.$$
We know from \cite{Ts Vl} that $A(q)=\sqrt{q}-1$ if $q$ is a
square.

The algebraic geometry bound \cite{Fen Lin Xin} says that
for a prime power $q$, one has
\begin{equation}\label{eq3.4}
\alpha_q^Q(\delta)\ge
1-2\delta-\frac{2}{A(q)}.\end{equation}

In the following part, we prove the bound (\ref{eq3.4}) for $\delta$ in the range $(0,1/2-2/A(q)-\log_q(1+2/q)$ using our result introduced in the previous sections.

{\it Proof of the bound (\ref{eq3.4})}
\begin{proof}
Let $\{\X/\F_q\}$ be a family of curves such that
$g(\X)\rightarrow\infty$ and $N(\X)/g(\X)\rightarrow A(q)$.

For $0<\delta<1/2-2/A(q)-\log_q(1+2/q)$, define two families of integers
$\{n=N(\X)\}_{\X}$ and
$\{m=\lfloor\delta N(\X)\rfloor+2g\}_{\X}$. Then
$n/g(\X)\rightarrow A(q)$ and $(m-2g)/n\rightarrow \delta$.

For each curve, let $P_1,\dots, P_n$ be $n$ $\F_q$-rational points and  choose a divisor $G$ of degree $m$ such that ${\rm supp}(G)\cap\{P_1,\dots,P_n\}=\emptyset$.

By Proposition \ref{3.5}, from each curve $\X$ with sufficiently large genus in the family we
have a Hermitian self-orthogonal code over $\F_{q^2}$ with parameters $[n,n-2(m-g+1)]$  and dual distance at least $m-2g+2$. By Lemma \ref{5.1}, we obtain a $q$-ary quantum $((n, q^{n-2(m-g+1)},m-2g+2))$ code. The desire bound follows.
\end{proof}

\end{document}